\documentclass[11pt]{article}

\usepackage{comment}

\usepackage[utf8]{inputenc}
\usepackage{fullpage}
\usepackage{amsmath}
\usepackage{amssymb}
\usepackage{amsthm}
\usepackage{mathtools}
\usepackage{hyperref}

\usepackage[algo2e,vlined,ruled,commentsnumbered]{algorithm2e}
\SetAlCapSkip{0.5em}

\usepackage[inline]{enumitem}

\usepackage{xcolor}

\usepackage{array}
\usepackage{float}

\usepackage{xfrac}

\usepackage[normalem]{ulem}

\newtheorem{theorem}{Theorem}

\newtheorem{fact}{Fact}
\newtheorem{conjecture}{Conjecture}
\newtheorem{observation}{Observation}

\DeclareMathOperator{\Exp}{Exp}

\DeclareMathOperator{\supp}{supp}
\DeclareMathOperator{\poly}{poly}

\title{\Large Simpler and Stronger Approaches for Non-Uniform Hypergraph Matching and the F{\"{u}}redi, Kahn, and Seymour Conjecture\thanks{This project received funding from the European Research Council (ERC) under the European Union's Horizon 2020 research and innovation programme (grant agreement No 817750) and from the Swiss National Science Foundation under grant 200021\_184622.
}}
\author{Georg Anegg\thanks{ETH Zurich. Email: ganegg@ethz.ch} \and Haris Angelidakis\thanks{TU Eindhoven. Email: c.angelidakis@tue.nl. Research was conducted while the author was at ETH Zurich.} \and Rico Zenklusen\thanks{ETH Zurich. Email: ricoz@ethz.ch}}
\date{}

\begin{document}

\maketitle

\begin{abstract}

A well-known conjecture of F{\"{u}}redi, Kahn, and Seymour (1993) on non-uniform hypergraph matching states that for any hypergraph with edge weights $w$, there exists a matching $M$ such that the inequality $\sum_{e\in M} g(e) w(e) \geq \mathrm{OPT}_{\mathrm{LP}}$ holds with $g(e)=|e|-1+\sfrac{1}{|e|}$, where $\mathrm{OPT}_{\mathrm{LP}}$ denotes the optimal value of the canonical LP relaxation.
While the conjecture remains open, the strongest result towards it was recently obtained by Brubach, Sankararaman, Srinivasan, and Xu (2020)---building on and strengthening prior work by Bansal, Gupta, Li, Mestre, Nagarajan, and Rudra (2012)---showing that the aforementioned inequality holds with $g(e)=|e|+O(|e|\exp(-|e|))$.
Actually, their method works in a more general sampling setting, where, given a point $x$ of the canonical LP relaxation, the task is to efficiently sample a matching $M$ containing each edge $e$ with probability at least $\sfrac{x(e)}{g(e)}$.

We present simpler and easy-to-analyze procedures leading to improved results.
More precisely, for any solution $x$ to the canonical LP, we introduce a simple algorithm based on exponential clocks for Brubach et al.'s sampling setting achieving $g(e)=|e|-(|e|-1)x(e)$. 
Apart from the slight improvement in $g$, our technique may open up new ways to attack the original conjecture. 
Moreover, we provide a short and arguably elegant analysis showing that a natural greedy approach for the original setting of the conjecture shows the inequality for the same $g(e)=|e|-(|e|-1)x(e)$ even for the more general hypergraph $b$-matching problem.

\end{abstract} 

\section{Introduction}

The maximum matching problem is one of the most heavily studied problems in Combinatorial Optimization. It has a natural and well-known generalization to hypergraphs, where the task is to find a maximum weight subset of non-overlapping (hyper-)edges. Formally, for a hypergraph $H=(V,E)$ with $E\subseteq 2^V$ and edge weights $w\in \mathbb{R}_{\geq 0}^E$, the task is to solve
\begin{equation*}
\max\{w(M)\colon M\subseteq E \text{ with } e_1\cap e_2=\emptyset \;\;\forall e_1,e_2\in M, e_1\neq e_2 \}\enspace,
\end{equation*}
where $w(M)\coloneqq \sum_{e\in M}w(e)$. (We use the shorthand  $x(U)\coloneqq \sum_{e\in U}x(e)$ for any $x\in\mathbb{R}_{\geq 0}^E$ and $U\subseteq E$.) Depending on the context, the above problem is called \emph{hypergraph matching} or \emph{set packing}, where the former is often used when deriving results as a function of the sizes of hyperedges.

Interestingly, despite intensive work there are still large gaps in our understanding of the approximability of hypergraph matching. In general hypergraphs the problem is equivalent to the maximum independent set problem, which means that it is $n^{1 - \varepsilon}$-hard to approximate, as shown by H{\aa}stad~\cite{hastad1999}. Nevertheless, significant progress has been achieved for $k$-uniform hypergraphs, i.e., when all edges have size $k$. (Note that the case $k=2$ corresponds to the classical matching problem.) In this setting, approximation guarantees are described as a function of $k$, and strong results have been obtained both with local search techniques~\cite{DBLP:journals/njc/Berman00,DBLP:conf/focs/Cygan13,DBLP:conf/iscopt/FurerY14} and approaches based on linear programming~\cite{CL12,FKS93,PP16} (see Section~\ref{sec:priorwork} below for additional information).

When dealing with non-uniform instances, which is the focus of this work, it is harder to identify a good fine-grained way to express approximation guarantees, as there is no obvious global parameter like $k$ in the case of $k$-uniform hypergraphs.\footnote{A simple way to parameterize the instance would be by $\max_{e\in E}|e|$, but this is very coarse as it does not allow for exploiting the existence of small hyperedges. Moreover, this parameterization essentially falls back to the $k$-uniform case as results for $k$-uniform hypergraph matching typically extend to hypergraphs with edges of size at most $k$.} Nevertheless, one would expect the problem to be easier to approximate if there are mostly hyperedges of small size. An approach to this problem was suggested almost three decades ago by F\"uredi, Kahn, and Seymour~\cite{FKS93}, who consider edge-wise guarantees with respect to a solution $x$ to the canonical LP relaxation, which is also called the \emph{fractional matching LP} and is defined as follows:
\begin{equation}
\renewcommand{\arraystretch}{1.3}
\begin{array}{c>{\displaystyle}r@{\ }c@{\ }ll}
\max          &\sum_{e \in E} w(e) x(e)\\
\textrm{s.t.} & x(\delta(v)) &\leq &1 & \forall v \in V\\[0.2em]
& x &\in & \mathbb{R}^E_{\geq 0}\enspace, & 
\end{array}\tag{\text{fractional matching LP}}
\end{equation}
where $\delta(v)\coloneqq\{e\in E\colon v\in e\}$ denotes all hyperedges containing $v$.
More precisely, F\"uredi et al.~\cite{FKS93} address the question of finding the ``smallest possible" function $g: E \to \mathbb{R}_{\geq 0}$ so that for every edge-weighted hypergraph $H=(V,E)$, there exists a matching $M$ satisfying
\begin{equation}\label{integralityGap}
\sum_{e\in M} g(e) w(e) \geq \sum_{e \in E} w(e) x(e)\enspace ,
\end{equation}
where $x \in [0,1]^E$ is an optimal solution to the fractional matching LP.
This leads to the following well-known conjecture, claiming that $g(e) = |e|-1+\frac{1}{|e|}$ is achievable.

\begin{conjecture}[F{\"{u}}redi et al.~\cite{FKS93}]\label{conj:fks}
For any hypergraph $H=(V,E)$ and edge weights $w\in \mathbb{R}^E_{\geq 0}$, there is a matching $M$ such that
\begin{equation*}
    \sum_{e \in M} \left( |e| - 1 + \frac{1}{|e|}  \right) w(e) \geq \sum_{e \in E} w(e)x(e) \enspace ,
\end{equation*}
where $x$ is an optimal solution to the frac\-tion\-al match\-ing LP.
\end{conjecture}
Even though the conjecture remains open, par\-tial progress has been achieved in several directions. In particular, F\"uredi et al.~\cite{FKS93} showed non-constructively that it holds if $H$ is either an
unweighted, uniform, or intersecting hypergraph.
The conjecture, if true, is best possible in the sense that it is known to be tight even for unweighted $k$-uniform hypergraphs for infinitely many values of $k$, namely whenever $k$ is one unit more than a prime power~\cite{FKS93}.

\medskip

A step towards the general form of Conjecture~\ref{conj:fks} was made by Bansal, Gupta, Li, Mestre, Nagarajan, and Rudra~\cite{BGLMNR12} in the context of the more general stochastic set packing problem. Although it is not stated explicitly in~\cite{BGLMNR12}, their algorithm shows that~(\ref{integralityGap}) holds for $g(e) = |e| + 1 + o(|e|)$. (A formal proof of this is given in~\cite{BBSX20}.) More precisely, they present a randomized procedure returning a matching that contains each edge $e$ with probability at least $\frac{x(e)}{|e| + 1 + o(|e|)}$, where $x$ is a fractional matching. Motivated by this, Brubach, Sankararaman, Srinivasan, and Xu~\cite{BBSX20} presented a strengthening of the algorithm in~\cite{BGLMNR12} to sample a matching containing each edge $e$ with probability at least $\frac{x(e)}{|e| + O(|e| \exp(-|e|))}$, which implies~(\ref{integralityGap}) for $g(e) = |e| + o(|e|)$. Inspired by these results, Brubach et al.~\cite{BBSX20} formulated the following slightly stronger conjecture.

\begin{conjecture}[Brubach et al.~\cite{BBSX20}] \label{conj:sampling}
For any hypergraph $H=(V,E)$, edge weights $w\in \mathbb{R}_{\geq 0}^E$, and fractional matching $x\in [0, 1]^E$, it is possible to efficiently sample a matching $M\subseteq E$ from a distribution such that each edge $e\in E$ is contained in $M$ with probability at least $\frac{x(e)}{|e|-1+\sfrac{1}{|e|}}$. 
\end{conjecture}
Such sampling settings are motivated by fairness considerations and have recently attracted significant attention in other contexts, such as in clustering (see, e.g.,~\cite{DBLP:journals/jmlr/HarrisLPST19,HPST19}).

Conjecture~\ref{conj:sampling} clearly implies Conjecture~\ref{conj:fks}. Moreover, it was observed in~\cite{BBSX20} that Conjecture~\ref{conj:fks} implies a non-constructive version of Conjecture~\ref{conj:sampling}, i.e., that there exists a distribution over matchings containing each edge $e$ with probability at least $\frac{x(e)}{|e|-1+\sfrac{1}{|e|}}$. We highlight that by adjusting a well-known technique based on both duality and the ellipsoid method by Carr and Vempala~\cite{DBLP:journals/rsa/CarrV02}, one can derive that a constructive version of Conjecture~\ref{conj:fks}, where the matching $M$ has to be found efficiently, also implies Conjecture~\ref{conj:sampling}. (For completeness, we show this in Appendix~\ref{sec:appendix}.) Hence, from a theoretical perspective, a key difference between Conjecture~\ref{conj:fks} and Conjecture~\ref{conj:sampling} is the constructiveness. Nevertheless, procedures working directly in the sampling setting remain of interest, as they are typically significantly faster than procedures obtained by using the ellipsoid method through the duality argument.

\subsection{Our contribution}

One main contribution of this work is to improve on Brubach et al.'s~\cite{BBSX20} sampling result, through a significantly simpler procedure that allows for a very short and clean analysis, and leads to an improved factor as highlighted below. The algorithm and its analysis are presented in Section~\ref{sec:exp-clocks}.
\begin{theorem}\label{mainTheorem}
Let $H=(V,E)$ be a hypergraph, $w\in \mathbb{R}^E_{\geq 0}$, and let $x\in [0,1]^E$ be a fractional matching. Then, there is an efficient sampling procedure returning a matching $M$ that satisfies
\begin{equation*}
\Pr[e\in M] \geq \frac{x(e)}{|e|-(|e|-1)x(e)} \quad\forall e\in E\enspace.
\end{equation*}
\end{theorem}
Note that while Theorem~\ref{mainTheorem} does not prove Conjecture~\ref{conj:sampling}, it is not implied by it either, i.e., the two statements are incomparable. Nevertheless, it strengthens Brubach et al.'s result and, as we discuss in the conclusion, may open up new ways to attack Conjecture~\ref{conj:fks}. However, the arguably most important point is the very clean and simple underlying algorithm and analysis.

Inspired by the result of Parekh and Pritchard~\cite{PP16} on $k$-uniform hypergraph $b$-matching, we then study Conjectures~\ref{conj:fks} and \ref{conj:sampling} in the more general non-uniform hypergraph $b$-matching setting. We recall that a hypergraph $b$-matching problem consists of a hypergraph $H=(V,E)$, edge weights $w\in \mathbb{R}^E_{\geq 0}$, and vertex capacities $b\in \mathbb{Z}^V_{\geq 1}$, and the task is to find a maximum weight $b$-matching $M$, which is a subset of edges such that, for any $v\in V$, at most $b(v)$ edges in $M$ contain $v$.
Moreover, a fractional $b$-matching is a point in the polytope $\{x \in [0,1]^E: x(\delta (v)) \leq b(v) \;\;\forall v \in V\}$. We show that a natural greedy algorithm, which picks edges in order of decreasing weights, allows for deriving the following result. This part is discussed in Section~\ref{sec:greedy}.
\begin{theorem}\label{thm:greedy}
Given a hypergraph $H = (V,E)$, $w\in \mathbb{R}_{\geq 0}^E$, $b\in \mathbb{Z}_{\geq 1}$, and a fractional $b$-matching $x \in [0,1]^E$, the greedy algorithm returns a $b$-matching $M$ that satisfies
\begin{equation*}
    \sum_{e \in M} \left( |e| - (|e| - 1)x(e)  \right) w(e) \geq \sum_{e \in E} w(e)x(e) \enspace.
\end{equation*}
\end{theorem}

Again, we think that a key contribution of this result is the simplicity of the greedy algorithm and the way we analyze it, which leads to a concise proof.

We finally highlight that, by extending a result of Carr and Vempala~\cite{DBLP:journals/rsa/CarrV02} based on LP duality and the ellipsoid method, we can use Theorem~\ref{thm:greedy} to obtain an efficient sampling algorithm returning a $b$-matching that contains each edge $e\in E$ with probability at least $\frac{x(e)}{|e| - (|e| - 1)x(e)}$. Because the employed arguments are quite standard, we formalize this in Appendix~\ref{sec:appendix} as Theorem~\ref{thm:appendix}. Even though this allows our results for $b$-matchings to be used to obtain an alternative proof of Theorem~\ref{mainTheorem}, the resulting sampling algorithm is significantly more involved and slower (due to the use of the ellipsoid method) than our much shorter and elegant approach that we use to prove Theorem~\ref{mainTheorem}.

\subsection{Further discussion on prior work}\label{sec:priorwork}
We expand on the progress for $k$-uniform hypergraph matching that we briefly mentioned above.
This special case remains $\mathtt{APX}$-hard for any $k\geq3$ as shown by Kann~\cite{DBLP:journals/ipl/Kann91}. Nevertheless, local search techniques led to an exciting sequence of strong results, culminating in an approximation factor of $\frac{k+1}{3}+\varepsilon$ for unweighted graphs by Cygan~\cite{DBLP:conf/focs/Cygan13} (see also F\"urer and Yu~\cite{DBLP:conf/iscopt/FurerY14}) and $\frac{k+1}{2}+\varepsilon$ for weighted graphs by Berman~\cite{DBLP:journals/njc/Berman00}. This is contrasted by an approximation hardness of $\Omega(\sfrac{k}{\log k})$ by Hazan, Safra, and Schwartz~\cite{DBLP:journals/cc/HazanSS06}.

Moreover, for $k$-uniform hypergraph matching, the integrality gap of the fractional matching LP is essentially settled. More precisely, F{\"{u}}redi et al.~\cite{FKS93} showed, non-constructively, that the integrality gap is at most $k-1+\frac{1}{k}$, and almost 20~years later, Chan and Lau~\cite{CL12} gave an elegant algorithmic version of that result. We recall that a matching lower bound on the integrality gap was given for all $k$ that are one unit more than a prime power~\cite{FKS93}.
Furthermore, Parekh and Pritchard~\cite{PP16} showed that one can obtain an LP-relative $(k - 1 + \frac{1}{k})$-approximation for the more general $b$-matching problem on $k$-uniform hypergraphs. This is a non-trivial generalization, which may be surprising in view of the fact that many results on matchings easily generalize to $b$-matchings.

\subsection{Brief overview of techniques}
The main tool for obtaining Theorem~\ref{mainTheorem} is the so-called exponential clocks technique. It was previously used by Buchbinder, Naor, and Schwartz~\cite{DBLP:journals/siamcomp/BuchbinderNS18} in the context of the Multiway Cut problem, and a further application in the context of Dynamic Facility Location was found by An, Norouzi-Fard, and Svensson~\cite{DBLP:journals/talg/AnNS17}. At a high level, it refers to competing independent exponential random  variables, where an exponential clock wins a competition if it has the smallest value among all participating exponential clocks. Similar in spirit to the randomized rounding algorithms of Bansal et al.~\cite{BGLMNR12} and Brubach et al.~\cite{BBSX20}, where the first step is to independently sample a random variable for each edge based on its $x$-value, and then extract a matching based on these random variables, the exponential clocks algorithm first independently samples an exponential random variable for each edge and then picks an edge if it is the ``winner" among all edges in its neighborhood (i.e., the realization of its exponential variable is the smallest among its neighborhood). To the best of our knowledge, the first use of the technique on matching problems was by Bruggmann and Zenklusen~\cite{BZ19}, where they used it as a tool to demonstrate the existence of certain distributions over matchings on graphs (not hypergraphs) in the context of designing contention resolution schemes for matchings. (We provide more details on this in Section~\ref{sec:exp-clocks}.) 

In this work, we generalize the arguments of~\cite{BZ19} and obtain the improved result stated in Theorem~\ref{mainTheorem}. Besides being a very elegant and easy-to-analyze algorithm, we hope that it may open up new ways to attack Conjectures~\ref{conj:fks} and~\ref{conj:sampling}, as our analysis already shows that, loosely speaking, when applied to an extreme point of the fractional matching LP, it implies an ``average version" of these conjectures. We expand on this in the conclusion.

Regarding Theorem~\ref{thm:greedy}, we carefully analyze the natural greedy algorithm for the more general hypergraph $b$-matching problem. The analysis is based on a concise charging scheme that charges the LP-contribution of non-matching edges first to tight vertices and then to edges in the $b$-matching to achieve the same factor (function $g$) as the one in Theorem~\ref{mainTheorem}. Thus, it addresses a generalization of Conjecture~\ref{conj:fks} in the $b$-matching setting.

\section{The exponential clocks rounding scheme}\label{sec:exp-clocks}

In this section, we present our algorithm based on exponential clocks that proves Theorem~\ref{mainTheorem}. For $\lambda >0$, we denote by $\Exp(\lambda)$ an exponential random variable with parameter $\lambda$. Our algorithm is described below as Algorithm~\ref{alg:exp-clocks}. Throughout this section, we use the notation $N(e)\coloneqq \{f\in E \colon f \neq e \textrm{ and } f\cap e \neq \emptyset \}$ for the edges overlapping with an edge $e$ and, for a fractional matching $x \in \mathbb{R}_{\geq 0}^E$, we denote its support by $\supp(x)\coloneqq \{e\in E\colon x(e) > 0\}$.
\begin{algorithm2e}
\vspace{0.1em}
Let $Z_e \sim \Exp(x(e))$ for $e\in \supp(x)$ be indep. random variables and $z_e$ a realization of $Z_e$\;

\Return{$M = \{e \in \supp(x): z_e < z_f \;\;\forall f \in N(e) \cap \supp(x)\}$}\;

\caption{Exponential clocks rounding}
\label{alg:exp-clocks}
\end{algorithm2e}

In words, our algorithm realizes for each edge an exponential random variable with parameter equal to its $x$-value, and picks each edge whose exponential random variable realized to a value smaller than any of its overlapping edges.
As already noted, our algorithm is a generalization of a procedure used in~\cite{BZ19} (see proof of Lemma 15 in~\cite{BZ19}) for a different purpose in the context of contention resolution schemes for matchings on graphs.\footnote{More precisely,~\cite{BZ19} uses Algorithm~\ref{alg:exp-clocks} for classical graphs $G=(V,E)$ to show that, for any $y\in \mathbb{R}_{>0}^E$, the vector $\left(\sfrac{y(e)}{(y(e) + \sum_{f \in N(e)} y(f))} \right)_{e \in E}$ is in the matching polytope, by proving that Algorithm~\ref{alg:exp-clocks} returns a matching $M\subseteq E$ with $\Pr[e \in M] = \sfrac{y(e)}{(y(e) + \sum_{f \in N(e)} y(f))}$ for all $e\in E$.}

To show that Algorithm~\ref{alg:exp-clocks} proves Theorem~\ref{mainTheorem}, we recall the following basic facts about exponential random variables, which lie at the heart of procedures based on exponential clocks.
\begin{fact}\label{expmin}
Let $X_i \sim \Exp(\lambda_i)$ for $i\in \{1,\ldots, n\}$ be indep. random variables. Then $\min\{X_1,\ldots, X_n\} \sim \Exp(\lambda_1+\ldots +\lambda_n)$.
\end{fact}
\begin{fact}\label{expcomp}
Let $X_1 \sim \Exp(\lambda_1), X_2\sim \Exp(\lambda_2)$ be indep. random variables. Then $\Pr[X_1< X_2]=\frac{\lambda_1}{\lambda_1+\lambda_2}$.
\end{fact}

\begin{proof}[Proof of Theorem~\ref{mainTheorem}]
Algorithm~\ref{alg:exp-clocks} clearly returns a matching $M$, and for each $e\in E$ with $x(e)=0$ we trivially have $\Pr[e\in M]=0=\sfrac{x(e)}{(|e|-(|e|-1)x(e))}$. Hence, it remains to check the probability of $M$ containing an edge $e\in \supp(x)$, which is equal to 
\begin{align*}
\Pr[e\in M] &= \Pr \left[Z_e < \min \{ Z_f \colon f \in N(e) \cap \supp(x)\} \right]\\
& = \Pr \left[Z_e < \Exp\left(x(N(e))\right) \right]\\
& = \frac{x(e)}{x(e) + x(N(e))}\enspace,
\end{align*}
where the first equality holds by construction, the second one holds due to Fact~\ref{expmin}, and the third one holds due to Fact~\ref{expcomp}. The result now follows by observing that
\begin{equation*}
x(N(e)) \leq \sum_{v\in e} \sum_{\substack{ f \in \delta(v):\\ f \neq  e }}x(f) \leq \sum_{v\in e} (1-x(e)) = |e|(1-x(e))\enspace,
\end{equation*}
where the first inequality is a union bound, and the second inequality holds since $x$ is a feasible fractional matching (and thus $x(\delta(v))\leq 1$).
\end{proof}

 \section{The greedy algorithm for non-uniform hypergraph \texorpdfstring{$b$}{b}-matching}\label{sec:greedy}

In this section, we switch to the more general hypergraph $b$-matching problem, and prove Theorem~\ref{thm:greedy} by analyzing a natural greedy algorithm with respect to the fractional hypergraph $b$-matching LP. A description of the greedy procedure is given in Algorithm~\ref{alg:greedy} below.
\begin{algorithm2e}[ht]

Order the edges $E=\{e_1,\ldots, e_m\}$ such that $w(e_1) \geq w(e_2) \geq \ldots \geq w(e_m)$\;

$M=\emptyset$\;

\For{$j=1$ to $m$}{
  \If{$M\cup \{e_j\}$ is a $b$-matching}{
  $M = M\cup\{e_j\}$\;
}
}
\Return{$M$}\;

\caption{Greedy algorithm for hypergraph $b$-matching}
\label{alg:greedy}
\end{algorithm2e}

To analyze Algorithm~\ref{alg:greedy} we present a charging scheme showing that the returned matching $M$ satisfies the properties of Theorem~\ref{thm:greedy}, and thus implies the theorem. For each edge $e\in E\setminus M$, the charging scheme first assigns its LP-weight $w(e)x(e)$ (where $x\in [0,1]^E$ is a fractional $b$-matching) to a vertex $v\in V$ that is \emph{saturated} by $M$, which means that the number of edges in $M$ containing $v$ is equal to $b(v)$. In a second step, the charge of saturated vertices is assigned to edges $f\in M$ in a way that $w(f) (|f|-(|f|-1)x(f))$ is large enough to pay for all charges assigned to it, as well as its own charge $w(f)x(f)$.
\begin{proof}[Proof of Theorem~\ref{thm:greedy}] First, it is clear that
Algorithm~\ref{alg:greedy} always returns a $b$-matching. Let $S = \{v \in V: |\delta(v)\cap M| = b(v)\}$ be the set of all vertices saturated by $M$. Note that the only reason for an edge $f$ not to be added to $M$, i.e., $f \in E \setminus M$, is that $f$ was ``blocked" by a vertex $v_f \in f \cap S$ that was already saturated when $f$ was considered by the algorithm. If $f$ contains more than one such saturated vertex at the moment $f$ is considered by the algorithm, we denote by $v_f$ an arbitrary vertex among those.
Since $v_f$ was already saturated when Algorithm~\ref{alg:greedy} considered $f$ and edges are considered in non-increasing order of weight, every edge in the final set $M$ that contains $v_f$ is at least as heavy as $f$, i.e.,
\begin{equation}\label{eqn:satdomination}
w(f) \leq w(e) \quad \forall e\in \delta(v_f) \cap M\enspace.
\end{equation}

We use this fact to show that the total  LP-contribution of edges in $E\setminus M$ can be bounded by an expression involving the $x$-load on saturated vertices and the weights of edges in $M$ containing them. In particular, we have
\begin{equation}\label{eqn:vertexcharging}
\begin{aligned}
\sum_{f\in E \setminus M} w(f)x(f)
 &= \sum_{v\in S} \sum_{\substack{f\in E\setminus M: \\ v_f = v}} w(f)x(f)\\
&\leq \sum_{v\in S} \min\{w(e)\colon e\in \delta(v) \cap M\} \sum_{\substack{f\in E\setminus M:\\v_f=v}} x(f)\\
&\leq \sum_{v\in S} \min\{w(e)\colon e\in \delta(v) \cap M\} \sum_{f\in \delta(v)\setminus M} x(f)\enspace,
\end{aligned}
\end{equation}
where the (first) equality can be interpreted as assigning, for each $f\in E\setminus M$, the LP-weight $w(f)x(f)$ to the vertex $v_f$ and the first inequality follows from~\eqref{eqn:satdomination}. We now turn to the last term of the right-hand side of the inequality above, the $x$-load of edges at a vertex $v\in S$ that are not in $M$, and we have
\begin{equation}\label{eqn:bmatching}
\sum_{f\in \delta(v)\setminus M} x(f)
 \leq b(v)-\sum_{e\in \delta(v)\cap M} x(e)= \sum_{e \in \delta(v)\cap M} (1 - x(e)) \enspace,
\end{equation}
where the inequality is due to $x$ being a feasible fractional $b$-matching (and thus $x(\delta(v))\leq b(v)$) and the equality is a consequence of $v\in S$ being saturated by $M$, i.e., $|\delta(v)\cap M| = b(v)$.

Finally, by combining~\eqref{eqn:vertexcharging} and~\eqref{eqn:bmatching} we obtain\begin{align*}
\sum_{f\in E \setminus M} w(f)x(f)
&\leq \sum_{v\in S}  \min\{w(e) \colon e\in \delta(v) \cap M\} \sum_{e\in \delta(v)\cap M}(1-x(e))\\
&\leq \sum_{v\in S} \sum_{e\in \delta(v)\cap M} w(e) (1-x(e))\\
& = \sum_{e\in M} |S \cap e|\, w(e)(1-x(e))\\
&\leq \sum_{e\in M} |e|\,  w(e)(1-x(e)) 
\enspace.
\end{align*}

By adding the term $\sum_{e\in M}w(e)x(e)$ to both sides of the above relation, the desired result is obtained.
\end{proof}

 \section{Conclusion}\label{sec:discussion}

The main goal of this work was the development of concise algorithms and analysis techniques for the F\"uredi, Kahn, and Seymour Conjecture and the related sampling problem, i.e., Conjectures~\ref{conj:fks} and \ref{conj:sampling}.
Although our techniques fall short of proving these two conjectures, they lead to the currently best bounds of $g(e)=|e|-(|e|-1) x(e) \leq |e|$ while being simpler than prior approaches.

We now briefly discuss why our factor $g$, and in particular the term $(|e| - 1)x(e)$, may open up new ways to attack both conjectures. For simplicity, let us focus on the original conjecture of F\"uredi, Kahn, and Seymour, i.e., Conjecture~\ref{conj:fks}. More precisely, notice that for any edge $e$ with $x(e)\geq \sfrac{1}{|e|}$, our bound is at least as good as the one conjectured by F\"uredi et al., because in this case $g(e)=|e|-(|e|-1)x(e)\leq |e|-1+\sfrac{1}{|e|}$. Hence, the contribution from these edges corresponds to the factor in the conjecture. Thus, the only issue is edges with small $x$-value. Of course, there are points in the matching polytope with only small $x$-values. However, this is not the case when $x$ is an extreme point of the fractional matching polytope. (Note that an optimal solution to the fractional matching LP, as stipulated by Conjecture~\ref{conj:fks}, can of course always be chosen to be an extreme point.\footnote{Moreover, also in the sampling setting, it suffices to prove Conjecture~\ref{conj:sampling} for extreme points only, as any fractional matching $x$ can first be written as a convex combination $\sum_{i=1}^q \lambda_i x^i$ of extreme points $x^i$ of the fractional matching polytope. If we can get a sampling procedure fulfilling the requirements of Conjecture~\ref{conj:sampling} for extreme points, then we can simply first pick with probability $\lambda_i$ the extreme point $x^i$ and then return a matching using the sampling procedure for $x^i$. One can easily check that such a sampling procedure fulfills the conditions of Conjecture~\ref{conj:sampling} if the sampling procedure for the extreme points $x^i$ does so.})

Indeed, if $x$ is an extreme point of the fractional matching polytope, then the average $x(e)$-value of an edge $e\in \supp(x)$ is at least $\sfrac{1}{|e|}$. This can be derived by using standard sparsity arguments as follows. Let $Q\coloneqq \{v\in V: x(\delta(v)) = 1\}$ be the set of all vertices for which the corresponding constraint in the fractional matching polytope is tight with respect to $x$. The remaining constraints of the fractional matching polytope are non-negativity constraints, which implies by classical sparsity arguments that $|\supp(x)|\leq |Q|$. Moreover, as all vertices in $Q$ are covered by one unit of $x$-value we have $|Q|\leq \sum_{e\in \supp(x)}|e| x(e)$. Hence, $1 \leq \frac{1}{|\supp(x)|}\sum_{e\in \supp(x)} |e| x(e)$, which can be interpreted as follows: ``on average'' the term $|e| x(e)$ is at least $1$, i.e., the ``average load'' of edges $e\in \supp(x)$ is at least $\sfrac{1}{|e|}$.

Of course, this is only an averaging reasoning and many edges $e$ will have $x$-value below $\sfrac{1}{|e|}$. Nevertheless, because $x$ is a maximum weight fractional matching, there is hope that edges with above-average $x$-value tend to be heavy edges, and the improved guarantee we get on those, compared to the guarantee needed by Conjecture~\ref{conj:fks}, may compensate for edges with small $x$-value. Another idea that may be helpful, possibly in combination with the above observation, is to perform a careful alteration of $x$ before applying our rounding procedure (which only requires $x$ to be a feasible, but not necessarily optimal fractional matching) with the goal to obtain stronger guarantees.

Finally, we think it is an interesting open question whether our exponential clocks algorithm can be extended to $b$-matching to obtain an efficient and hopefully simple sampling procedure for $b$-matchings.
 
\section*{Acknowledgments} The authors would like to thank David Harris for suggesting the elegant formulation of Theorem~\ref{thm:general_carr-vempala} in the Appendix.

 \appendix
\section{Relating Conjectures~\ref{conj:fks} and \ref{conj:sampling} via LP duality}\label{sec:appendix}

In this section, we prove the following natural generalization of The\-o\-rem~\ref{mainTheorem} to hypergraph $b$-matching.

\begin{theorem}\label{thm:appendix}
Let $H=(V,E)$ be a hypergraph, $b\in \mathbb{Z}^V_{\geq 1}$, $w\in \mathbb{R}^{E}_{\geq 0}$, and let $x\in [0,1]^E$ be a fractional $b$-matching. Then, there is an efficient sampling procedure returning a $b$-matching $M$ that satisfies
\begin{equation*}
    \Pr[e \in M] \geq \frac{x(e)}{|e|-(|e|-1)x(e)} \quad \forall e\in E\enspace.
\end{equation*}
\end{theorem}

In fact, we show that Theorem~\ref{thm:appendix} is implied by Theorem~\ref{thm:greedy} via a black-box application of a general result, stated as Theorem~\ref{thm:general_carr-vempala} below. Thus, this leads to an alternative algorithm for Theorem~\ref{mainTheorem}, which is slower but still polynomial-time. In the statement below, we denote by $\langle p \rangle$ the encoding length of $p$.
\begin{theorem}\label{thm:general_carr-vempala}
Let $\mathcal S \subseteq 2^U$ be a collection of subsets over a finite ground set $U$, and let $p: U \to [0,1]$. Suppose that there is an algorithm $\mathcal{A}$ that for any weight function $w: U \rightarrow \mathbb{R}$ computes a set $S \in \mathcal S$ such that
\begin{equation*}
\sum_{u \in S} w(u) \geq \sum_{u \in U} p(u) w(u) \enspace.
\end{equation*}
Then, there is a procedure doing $\poly(|U|, \langle p \rangle)$ many calls to $\mathcal{A}$ and performing further operations taking $\poly(|U|, \langle p \rangle)$ time that outputs a collection $\mathcal{S}' \subseteq \mathcal{S}$ of sets, along with coefficients $\{\lambda_S\}_{S \in \mathcal{S}'}$ that satisfy $\lambda_S \geq 0$ for every $S \in \mathcal{S}'$ and $\sum_{S \in \mathcal{S}'} \lambda_S = 1$, such that
\begin{equation}\label{eq:goodDist}
\sum_{S \in \mathcal{S}': u \in S}  \lambda_S \geq p(u) \quad \forall u\in U\enspace.
\end{equation}
\end{theorem}
We highlight that distributions obtained through Theorem~\ref{thm:general_carr-vempala} have support bounded by $\poly(|U|,\langle p \rangle)$. In particular, this follows from the fact that the procedure only does $\poly(|U|,\langle p \rangle)$ many operations and calls to algorithm $\mathcal{A}$.\footnote{Moreover, standard techniques allow for efficiently reducing any polynomial-support distribution satisfying~\eqref{eq:goodDist} to one whose support has size at most  $|U|+1$. This follows from the fact that~\eqref{eq:goodDist} consists of only $|U|$-many linear inequalities.}

Theorem~\ref{thm:general_carr-vempala} builds on the seminal work of Carr and Vempala~\cite{DBLP:journals/rsa/CarrV02} that deals with a uniform bound on the integrality gap of linear programming relaxations of integer programs, whereas the guarantees we are seeking are element-specific, and thus, do not follow from the result of~\cite{DBLP:journals/rsa/CarrV02}. Even though the generalization is quite straightforward, we include it here for completeness and because it seems that this strong connection has been overlooked in prior work.

Before proving Theorem~\ref{thm:general_carr-vempala}, we show how we can use it along with Theorem~\ref{thm:greedy} to prove Theorem~\ref{thm:appendix}.

\begin{proof}[Proof of Theorem~\ref{thm:appendix}]
We will invoke Theorem~\ref{thm:general_carr-vempala} for suitable choices of $U, \mathcal{S}$ and $p$, and with the greedy algorithm as $\mathcal{A}$.

More precisely, let $U=E$, let $\mathcal{S}$ be the set of $b$-matchings of $H$, and $p(e)=\frac{x(e)}{g(e)}$ for every $e \in E$, where $g(e)=|e|-(|e|-1)x(e)$ (note that $p(e) \in [0,1]$, as $x(e) \in [0,1]$ and $g(e) \geq 1$ for every $e \in E$). We claim that with the greedy algorithm (see Algorithm~\ref{alg:greedy}) as procedure $\mathcal{A}$, the assumption of Theorem~\ref{thm:general_carr-vempala} is satisfied.
Then, since the greedy algorithm is efficient, the conclusion of Theorem~\ref{thm:general_carr-vempala} gives us the desired distribution, which leads to an efficient sampling procedure.

It remains to prove that for any $w:E\to \mathbb{R}$, the greedy algorithm can be used to compute a $b$-matching $M\in \mathcal S$ that satisfies
\begin{equation}\label{eqn:assumption}
\sum_{e \in M} w(e) \geq \sum_{e \in E} p(e) w(e) \enspace.
\end{equation}
To this end, consider the adjusted weight function $\overline{w}(e) \coloneqq \max \left\{\frac{w(e)}{g(e)},0 \right\}$ for $e \in E$. Running Algorithm~\ref{alg:greedy} with weight function $\overline{w}$, we obtain a $b$-matching $M'\in \mathcal S$ that, by Theorem~\ref{thm:greedy}, satisfies
\begin{equation*}
\sum_{e\in M'} g(e)\overline{w}(e) \geq \sum_{e\in E} \overline{w}(e)x(e) \enspace.
\end{equation*}
Let $M \coloneqq \{e\in M' : \overline{w}(e) > 0\}$ be the $b$-matching consisting of all edges in $M$ of strictly positive $\overline{w}$-weight. Hence, because $\overline{w}(e)= \max \left\{\frac{w(e)}{g(e)},0 \right\}$, we have
\begin{equation*}
\sum_{e\in M} g(e) \frac{w(e)}{g(e)} \geq \sum_{e\in E} \frac{w(e)}{g(e)}x(e) \enspace.
\end{equation*}
This is equivalent to~\eqref{eqn:assumption}, and thus completes the proof.
\end{proof}

The reasoning above does not depend on the precise form of $g$ (as long as $g(e)\geq 1$ for all $e\in E$). This means that the proof still holds if $g$ is replaced by another function, as long as one can obtain a result analogous to Theorem~\ref{thm:greedy} with $|e|-(|e|-1)x(e)$ replaced by $g(e)$; in particular, it shows that the existence of an efficient algorithm returning a matching satisfying Conjecture~\ref{conj:fks} implies Conjecture~\ref{conj:sampling}.

\medskip

We now prove Theorem~\ref{thm:general_carr-vempala}. The main tools used are LP duality and the ellipsoid method.

\begin{proof}[Proof of Theorem~\ref{thm:general_carr-vempala}] We start by writing an exponential-size linear program that determines whether there exists a collection $\mathcal{S'} \subseteq \mathcal{S}$ of sets, not necessarily polynomially-sized, along with corresponding coefficients $\{\lambda_S\}_{S \in \mathcal{S}'}$ such that $\sum_{S \in \mathcal{S}': u \in S} \lambda_S \geq p(u)$ for every $u \in U$. We call this the primal LP, and it is easy to see that its objective value is $0$, whenever it is feasible, and moreover, it is feasible if and only if such a collection $\mathcal{S}'$ of sets exists.
\begin{equation*}
\everymath={\displaystyle}
\renewcommand{\arraystretch}{1.3}
\begin{array}{cr@{\ }c@{\ }ll}
\min          &0 \quad\quad\quad\\
\textrm{s.t.} & \sum_{S \in \mathcal{S}: u \in S} \lambda_S &\geq & p(u) & \forall u \in U\\
& \sum_{S \in \mathcal{S}} \lambda_S &=& 1\\
& \lambda &\in & \mathbb{R}^{\mathcal S}_{\geq 0} \enspace. &
\end{array}
\end{equation*}
Its dual is shown below.
\begin{equation*}
\everymath={\displaystyle}
\renewcommand{\arraystretch}{1.3}
\begin{array}{cr@{\ }c@{\ }ll}
\max          & \multicolumn{3}{l}{ \sum_{u \in U} p(u)   y(u) - \mu}  \\
\textrm{s.t.} & y(S) &\leq & \mu & \forall S \in \mathcal{S}\\
& \mu &\in& \mathbb{R}\\
& y &\in &\mathbb{R}^U_{\geq 0} \enspace. &
\end{array}
\end{equation*}

The dual LP is feasible because of the all-zeros vector, which leads to an objective value of $0$. Also observe that the dual LP is scale-invariant, in the sense that if $(y,\mu)$ is feasible, then so is $(\gamma y, \gamma \mu)$ for any $\gamma \geq 0$. Thus, if there is any dual LP solution of strictly positive objective value, then the dual LP is unbounded. Consequently, the dual LP is unbounded if and only if there is a solution of value at least $1$, and we describe all those solutions by the following polyhedron.
\begingroup
\everymath={\displaystyle}
\renewcommand{\arraystretch}{1.3}
\begin{equation*}
   \mathcal{Q}(p) \coloneqq  \left\{
   (y, \mu)  \in \mathbb{R}_{\geq 0}^{U} \times \mathbb{R} \, \middle\vert
      \begin{array}{r@{\,}c@{\,}lr}
    \sum_{u \in U}  p(u)    y(u) - \mu &\geq  &1\\
     y(S)  &\leq &\mu  &\quad \forall S \in \mathcal{S}
   \end{array}
    \right \}\,.
\end{equation*}
\endgroup
The above discussion implies the following.
\begin{observation}\label{obs:appendix}
The dual LP has optimal value $0$ if and only if $\mathcal{Q}(p) = \emptyset$. Equivalently, the primal LP is feasible if and only if $\mathcal{Q}(p) = \emptyset$.
\end{observation}

We will use the ellipsoid method together with algorithm $\mathcal A$ to certify (with polynomially many calls to $\mathcal{A}$) that $\mathcal{Q}(p) = \emptyset$.

\paragraph{An ellipsoid iteration.} Let $(y, \mu) \in \mathbb{R}_{\geq 0}^{U} \times \mathbb{R}$ be a candidate point of $\mathcal{Q}(p)$. If $(y,\mu)$ violates the first constraint in the inequality description of $\mathcal{Q}(p)$, then we can use this inequality as a separating hyperplane. 
So, suppose that $(y,\mu)$ satisfies $\sum_{u \in U} p(u) y(u)- \mu \geq 1$. We now run algorithm $\mathcal A$ with weights $w\in \mathbb{R}^{U}_{\geq 0}$ given by $w(u) = y(u)$ for $u \in U$ to obtain a set $S\in \mathcal S$ that satisfies $ \sum_{u \in S} w(u) \geq \sum_{u \in U} p(u)w(u)$.
Since we have assumed that the first constraint of $\mathcal{Q}(p)$ is satisfied by $(y,\mu)$, we get
\begin{equation*}
    y(S) = w(S) \geq \mu + 1\enspace.
\end{equation*}
Since $S\in \mathcal S$, the above implies that $(y,\mu)$ violates the constraint $y(S) \leq \mu$, which we can thus use as a separating hyperplane. 
In short, for any candidate point $(y, \mu) \in \mathbb{R}_{\geq 0}^{U} \times \mathbb{R}$, we can find a constraint of $\mathcal{Q}(p)$ that is violated by $(y, \mu)$ with a single call to algorithm $\mathcal{A}$.
Thus we have a separation oracle and the ellipsoid method will certify that $\mathcal{Q}(p) = \emptyset$.
Since the encoding length of the separating hyperplanes generated is bounded by $\poly (|U|,\langle p \rangle)$, the number of ellipsoid iterations is polynomially bounded  (see Theorem 6.4.9 of~\cite{schrijver2012geometric}).

\paragraph{Computing the desired distribution.} Combined with Observation~\ref{obs:appendix}, the above discussion shows that the primal LP is feasible, and thus the distribution we are looking for exists. Such a distribution can now be found through standard techniques. More precisely, the ellipsoid method certified emptiness of $\mathcal{Q}(p)$, and did so with the polynomially many constraints that were generated through the separation oracle. Hence, even when replacing in $\mathcal{Q}(p)$ the family $\mathcal{S}$ by the family $\mathcal{S}'\subseteq \mathcal{S}$ of all sets $S \in \mathcal{S}$ that we computed in our calls to the separation oracle to obtain violated constraints of type $y(S) \leq \mu$, the corresponding polyhedron remains empty.
Hence, by strong duality, the primal LP is feasible even when replacing $\mathcal{S}$ by the polynomially-sized subset $\mathcal{S}'$. This reduced primal LP can now be solved efficiently and leads to a distribution with the desired properties.\end{proof}


\begin{thebibliography}{10}

\bibitem{DBLP:journals/talg/AnNS17}
H.{-}C. An, A.~Norouzi{-}Fard, and O.~Svensson.
\newblock Dynamic facility location via exponential clocks.
\newblock {\em {ACM} Transactions on Algorithms}, 13(2):21:1--21:20, 2017.

\bibitem{BGLMNR12}
N.~Bansal, A.~Gupta, J.~Li, J.~Mestre, V.~Nagarajan, and A.~Rudra.
\newblock When {LP} is the cure for your matching woes: Improved bounds for
  stochastic matchings.
\newblock {\em Algorithmica}, 63(4):733--762, 2012.

\bibitem{DBLP:journals/njc/Berman00}
P.~Berman.
\newblock A $d/2$ approximation for maximum weight independent set in $d$-claw
  free graphs.
\newblock {\em Nordic Journal of Computing}, 7(3):178--184, 2000.

\bibitem{BBSX20}
B.~Brubach, K.~A. Sankararaman, A.~Srinivasan, and P.~Xu.
\newblock Algorithms to approximate column-sparse packing problems.
\newblock {\em {ACM} Transactions on Algorithms}, 16(1):10:1--10:32, 2020.

\bibitem{BZ19}
S.~Bruggmann and R.~Zenklusen.
\newblock An optimal monotone contention resolution scheme for bipartite
  matchings via a polyhedral viewpoint.
\newblock {\em CoRR}, abs/1905.08658, 2019.

\bibitem{DBLP:journals/siamcomp/BuchbinderNS18}
N.~Buchbinder, J.~(Seffi) Naor, and R.~Schwartz.
\newblock Simplex partitioning via exponential clocks and the multiway-cut
  problem.
\newblock {\em {SIAM} Journal on Computing}, 47(4):1463--1482, 2018.

\bibitem{DBLP:journals/rsa/CarrV02}
R.~D. Carr and S.~S. Vempala.
\newblock Randomized metarounding.
\newblock {\em Random Structures and Algorithms}, 20(3):343--352, 2002.

\bibitem{CL12}
Y.~H. Chan and L.~C. Lau.
\newblock On linear and semidefinite programming relaxations for hypergraph
  matching.
\newblock {\em Mathematical Programming}, 135(1-2):123--148, 2012.

\bibitem{DBLP:conf/focs/Cygan13}
M.~Cygan.
\newblock Improved approximation for 3-dimensional matching via bounded
  pathwidth local search.
\newblock In {\em Proceedings of the 54th Annual {IEEE} Symposium on
  Foundations of Computer Science ({FOCS})}, pages 509--518, 2013.

\bibitem{FKS93}
Z.~F{\"{u}}redi, J.~Kahn, and P.~D. Seymour.
\newblock On the fractional matching polytope of a hypergraph.
\newblock {\em Combinatorica}, 13(2):167--180, 1993.

\bibitem{DBLP:conf/iscopt/FurerY14}
M.~F{\"{u}}rer and H.~Yu.
\newblock Approximating the $k$-set packing problem by local improvements.
\newblock In {\em Proceedings of the 3rd International Symposium on
  Combinatorial Optimization ({ISCO})}, pages 408--420, 2014.

\bibitem{schrijver2012geometric}
M.~Gr{\"o}tschel, L.~Lov{\'a}sz, and A.~Schrijver.
\newblock {\em Geometric Algorithms and Combinatorial Optimization}, 2nd edition.
\newblock Springer, 1993.

\bibitem{DBLP:journals/jmlr/HarrisLPST19}
D.~G. Harris, S.~Li, T.~W. Pensyl, A.~Srinivasan, and K.~Trinh.
\newblock Approximation algorithms for stochastic clustering.
\newblock {\em Journal of Machine Learning Research}, 20:153:1--153:33, 2019.

\bibitem{HPST19}
D.~G. Harris, T.~W. Pensyl, A.~Srinivasan, and K.~Trinh.
\newblock A lottery model for center-type problems with outliers.
\newblock {\em {ACM} Transactions on Algorithms}, 15(3):36:1--36:25, 2019.

\bibitem{hastad1999}
J.~H{\aa}stad.
\newblock Clique is hard to approximate within $n^{1-\varepsilon}$.
\newblock {\em Acta Mathematica}, 182(1):105--142, 1999.

\bibitem{DBLP:journals/cc/HazanSS06}
E.~Hazan, S.~Safra, and O.~Schwartz.
\newblock On the complexity of approximating $k$-set packing.
\newblock {\em Computational Complexity}, 15(1):20--39, 2006.

\bibitem{DBLP:journals/ipl/Kann91}
V.~Kann.
\newblock Maximum bounded 3-dimensional matching is {MAX} {SNP}-complete.
\newblock {\em Information Processing Letters}, 37(1):27--35, 1991.

\bibitem{PP16}
O.~Parekh and D.~Pritchard.
\newblock Generalized hypergraph matching via iterated packing and local ratio.
\newblock In {\em Proceedings of 12th International Workshop on Approximation
  and Online Algorithms ({WAOA})}, pages 207--223, 2014.
  
  


  

\end{thebibliography}
\end{document}